\documentclass[conference]{IEEEtran}
\IEEEoverridecommandlockouts
\usepackage{amsmath}
\usepackage{color}
\usepackage{subfigure}
\usepackage{array}
\usepackage{pifont}
\usepackage{{inputenc}}
\usepackage{balance}
\usepackage{multirow,tabularx}
\usepackage{booktabs,longtable}
\usepackage[ruled,vlined]{algorithm2e}
\usepackage{algpseudocode}
\usepackage{graphicx}
\usepackage{bbm}
\usepackage{enumitem}
\usepackage{verbatim}
\usepackage{amsfonts}
\usepackage{hyperref}
\usepackage{enumitem}
\usepackage{xcolor}
\usepackage{xspace}
\usepackage{multirow}
\usepackage{ntheorem}
\newtheorem{definition}{Definition}
\newtheorem{theorem}{Theorem}
\newtheorem*{proof}{Proof}
\usepackage{dsfont}

\newcommand{\var}{\mathop{\mathrm{Var}}}

\newcommand{\ie}{\emph{i.e., }}
\newcommand{\eg}{\emph{e.g., }}
\newcommand{\etal}{\emph{et al. }}

\newcommand{\etc}{\emph{etc.}}
\newcommand{\wrt}{\emph{w.r.t. }}

\hypersetup{
  colorlinks,
  linkcolor={blue!70!green},
  citecolor={green!70!blue},
  urlcolor={blue!70!green}
}

\usepackage{cleveref}

\newcommand{\mypara}[1]{\noindent\textbf{#1.} \xspace}

\newcommand{\B}[1]{\textbf{#1}}
\newcommand{\mymethod}{\ensuremath{\mathsf{RetraSyn}}\xspace}
\newcommand{\mymethodb}{\ensuremath{\mathsf{RetraSyn\textsuperscript{b}}}\xspace}
\newcommand{\mymethodu}{\ensuremath{\mathsf{RetraSyn\textsuperscript{p}}}\xspace}
\newcommand{\dmu}{\ensuremath{\mathsf{DMU}}\xspace}
\newcommand{\nosampling}{\ensuremath{\mathsf{AllUpdate\textsuperscript{b}}}\xspace}
\newcommand{\nosamplingu}{\ensuremath{\mathsf{AllUpdate\textsuperscript{p}}}\xspace}
\newcommand{\noeq}{\ensuremath{\mathsf{NoEQ\textsuperscript{b}}}\xspace}
\newcommand{\noequ}{\ensuremath{\mathsf{NoEQ\textsuperscript{p}}}\xspace}

\newcommand{\ldpids}{\ensuremath{\mathsf{LDP\text{-}IDS}}\xspace}

\def\BibTeX{{\rm B\kern-.05em{\sc i\kern-.025em b}\kern-.08em
    T\kern-.1667em\lower.7ex\hbox{E}\kern-.125emX}}
\begin{document}

\title{Real-Time Trajectory Synthesis with Local Differential Privacy
}

\author{
Yujia Hu$^{\dagger}$, Yuntao Du$^{\sharp}$, Zhikun Zhang$^{\ddag}$, Ziquan Fang$^{\dagger}$, Lu Chen$^{\dagger}$, Kai Zheng$^{\S}$, Yunjun Gao$^{\dagger}$\\
\normalsize $^{\dagger}$\emph{Zhejiang University, China}~~~~\hspace*{0.1in}
\normalsize $^{\sharp}$\emph{Purdue University, USA}~~~~\hspace*{0.1in}
\normalsize $^{\ddag}$\emph{Stanford University, USA}\\
\normalsize $^{\S}$\emph{University of Electronic Science and Technology of China, China}\\
$^{\dagger}$\emph{\{charliehu, zqfang, luchen, gaoyj\}@zju.edu.cn}~~\hspace*{0.1in} ~~$^{\sharp}$\emph{ytdu@purdue.edu}~~\hspace*{0.1in}$^{\ddag}$\emph{zhikun@stanford.edu}~~\hspace*{0.1in}$^{\S}$\emph{zhengkai@uestc.edu.cn} \\
}

\maketitle

\begin{abstract}
Trajectory streams are being generated from location-aware devices, such as smartphones and in-vehicle navigation systems.
Due to the sensitive nature of the location data, directly sharing user trajectories suffers from privacy leakage issues. 
\textit{Local differential privacy} (LDP), which perturbs sensitive data on the user side before it is shared or analyzed, emerges as a promising solution for private trajectory stream collection and analysis. 
Unfortunately, existing stream release approaches often neglect the rich spatial-temporal context information within trajectory streams, resulting in suboptimal utility and limited types of downstream applications. 
To this end, we propose \mymethod, a novel real-time trajectory synthesis framework, which is able to perform on-the-fly trajectory synthesis based on the mobility patterns privately extracted from users' trajectory streams. Thus, the downstream trajectory analysis can be performed on the high-utility synthesized data with privacy protection.  We also take the genuine behaviors of real-world mobile travelers into consideration, ensuring authenticity and practicality.
The key components of \mymethod include the global mobility model, dynamic mobility update mechanism, real-time synthesis, and adaptive allocation strategy.
We conduct extensive experiments on multiple real-world and synthetic trajectory datasets under various location-based utility metrics, encompassing both streaming and historical scenarios. 
The empirical results demonstrate the superiority and versatility of our proposed framework. 

\end{abstract}


\section{Introduction} 
\label{sec:introduction}

Massive trajectory streams are being generated from location-aware devices, such as GPS sensors in smartphones.
A \textit{trajectory stream}, which is formed by continuously reported locations, plays a crucial role in real-time applications such as traffic monitoring~\cite{traffic_monitoring}, emergency response~\cite{emergency_response}, location-based services~\cite{urbancomputing}, etc. For instance, in traffic management, trajectory streams continuously generated by vehicles can be used to monitor the dynamic traffic flow and perform real-time congestion prediction~\cite{traffic_monitoring}. 

While trajectory streams have great potential in real-life applications, privacy issues arise during trajectory stream collection. Since the data collector is not always trusted, sensitive information about an individual's locations may be revealed by adversaries~\cite{deanonymization,membership1,membership2,ccs18_adatrace}, which limits the practical applications of trajectory stream analysis.
\textit{Local differential privacy} (LDP) has emerged as the \textit{de facto} standard for private data collection with rigorous mathematical guarantees. Given users' sensitive information, LDP defines various randomized algorithms to perturb the original data, such that attackers cannot distinguish the individual inputs given the perturbed outputs.
Afterward, the perturbed outputs can be aggregated and published safely for downstream analysis. LDP provides a promising solution for data sharing without relying on any trusted data curator and thereby has been employed by many companies (\eg Google~\cite{google_rappor} and Microsoft~\cite{microsoft_nips}).

Recently, LDP has been applied to handle the general streaming data publication tasks~\cite{sigmod22_ldpids}. However, it treats trajectory streams as ordinary statistical streams and focuses on general statistical tasks, such as count queries and mean estimation. 
The general idea is to perturb each \textit{location} in the trajectory stream independently and publish aggregated statistics with histograms, which leaves the rich spatial-temporal information in trajectory streams unexplored. 
However, the spatial-temporal behaviors are indispensable for practical trajectory data release~\cite{vldb23_ldptrace,vldb21_ngram,vldb23_direction}, and the failure to preserve these characteristics significantly hampers its utility. Moreover, the above task-dependent solution cannot deal with arbitrary downstream location-based tasks, which limits its usability.

Researchers also developed other trajectory-aware frameworks~\cite{vldb23_ldptrace,vldb21_ngram,vldb23_direction} under LDP, which aim to capture the spatial-temporal patterns and publish complete trajectories as a safe substitute for the original dataset.
However, they typically perform one-time releases for \textit{historical trajectories}. In streaming scenarios, locations are sequentially reported; the reliance on \textit{historical features} (\eg full trajectory length~\cite{vldb23_direction,vldb23_ldptrace,vldb21_ngram}) prevents them from real-time processing. 

 The above concerns motivate us to propose \mymethod, an effective \textit{real-time} synthesis framework tailored to spatial-temporal trajectory streams with the protection of LDP. The challenges are three-fold.

 \textit{Challenge I: How to dynamically release trajectory streams for arbitrary downstream tasks?} Real-world location-based analysis tasks are not limited to rudimentary statistical queries. The complexity of many tasks necessitates the extraction of spatial-temporal patterns inherent in trajectory streams. Simultaneously, the imperative of real-time responsiveness brings additional challenges. Existing solutions either ignore the spatial-temporal context in trajectory streams or are not applicable in streaming scenarios, which fail to achieve these two features concurrently. To this end, \mymethod constructs a global mobility model by aggregating users' perturbed \textit{transition states} at each timestamp, which takes the continuous movement patterns into consideration. It also employs a synthesis-based framework to dynamically generate synthetic trajectories that align with the current learned spatial-temporal patterns. Thus, our \mymethod is able to perform real-time releasing and achieve high versatility in solving various downstream tasks with rigorous statistical privacy.

\textit{Challenge II: How to accurately capture the dynamic spatial-temporal patterns in a real-time manner?}  
 Unlike existing trajectory-aware frameworks that consider static historical trajectories, locations are continuously reported in real-time scenarios and should also be published sequentially. 
 Moreover, for high-dimensional data release, different dimensions may exhibit different trends. Existing LDP streaming data publish strategies~\cite{sigmod22_ldpids} treat dimensions equally, neglecting such fine-grained dynamics, which leads to unnecessary perturbation noise.
To facilitate real-time processing, a dynamic mobility update (\dmu) mechanism is proposed to adjust the global mobility model on the fly. 
By evaluating perturbation noise and the dimensional diversity, \dmu employs an optimization-based strategy to selectively update the most informative parts of the mobility model at each timestamp.
Through its ability to recognize dimension differences, the \dmu mechanism can accurately capture the changing patterns (\ie significant transitions) and reduce perturbation noise.

\textit{Challenge III: How to better preserve utility and authenticity in the dynamic setting?} 
In real-world scenarios, the number of trajectory streams may vary over time (users turn on/off the tracking devices), and the entering/quitting status of each trajectory is also undetermined. This brings unique challenges in generating authentic trajectories that reflect these real-world dynamics. Additionally, the intricate dynamics of real-world trajectory streams also bring challenges to allocation strategy, which is also a crucial part of streaming data release. Existing LDP streaming allocation strategies~\cite{sigmod22_ldpids} rely on a fixed active user set to determine the appropriate size of report users, which is impractical in this realistic setting and hampers its utility.  To mitigate this, \mymethod integrates entering and quitting events into our global mobility model to emulate the behaviors of genuine users. In our allocation strategies, we maintain a dynamic active user set and explore new portion-based allocation approaches to ensure sufficient utilization of the privacy budget/report users under realistic dynamic scenarios.


In summary, the \mymethod framework consists of four components: global mobility model, dynamic mobility update (\dmu), real-time synthesis, and adaptive allocation strategy. 
At each timestamp, users' mobility patterns are perturbed and aggregated to construct the global mobility model. Afterward, the \dmu mechanism selects \textit{significant patterns} and dynamically updates the global mobility model. Finally, a Markov-based probabilistic model is applied to generate synthetic trajectories that align with the current updated spatial-temporal patterns. To appropriately distribute the privacy budget/report users to each timestamp, we introduce different portion-based adaptive allocation strategies based on budget division and population division.
Our real-time trajectory synthesis process is locally differentially private, meaning that the global moving patterns (as well as the synthetic trajectories) are not strongly dependent on any specific user at any timestamp. 

The main contributions of this paper are four-fold:
\begin{itemize}[leftmargin=*]
    \item We propose \mymethod, the \textit{first} locally differentially private trajectory synthesis framework designed for trajectory streams.
    It can perform real-time trajectory generation based on mobility patterns while protecting users' sensitive data.
    \item We construct an effective global mobility model to capture the complex spatial-temporal contexts inherent in trajectory streams, and propose a dynamic mobility update mechanism to mimic the dynamics of real trajectories over time. 
    \item To synthesize realistic trajectories, we explore the entering and quitting behaviors of traces to better capture the evolution of the trajectory stream. 
    Different privacy budget allocation strategies are considered under complex real-world scenarios where active users vary across timestamps.
    \item We conduct comprehensive experiments on both synthetic and real-world datasets. Our evaluation metrics encompass both streaming data analysis and history data analysis. 
    The results demonstrate our superior performance and versatility.
\end{itemize}

\mypara{Roadmap}
The paper is structured as follows.
\autoref{sec:preliminaries} provides foundational knowledge on LDP and the streaming setting. 
\autoref{sec:method} introduces the technical details of \mymethod and theoretical analysis is presented in \autoref{sec:analysis}. 
The empirical evaluation of \mymethod is in \autoref{sec:experiments}. \autoref{sec:related_work} reviews existing literature on stream release and historical trajectory publication.  We finally conclude the paper in \autoref{sec:conclusion}.

\section{Preliminaries} \label{sec:preliminaries}
In this section, we first introduce the concepts of LDP and LDP for streaming data. Then, we formalize our problem.

\subsection{Local Differential Privacy}

Local differential privacy (LDP) is a strong privacy-preserving paradigm that offers a provable mathematical guarantee. In LDP, a \textit{data curator} seeks to collect sensitive information from a large number of \textit{users}. To preserve privacy, each user independently applies local perturbations to their data values and subsequently reports the noisy output to the curator for aggregation. 

Formally, let $\Psi$ be a randomized mechanism that takes the original data value $x$ of each user and outputs the perturbed value $\Psi(x)$. The $\epsilon$-LDP notion is formulated as follows:
\begin{definition}[$\epsilon$-Local Differential Privacy]
An algorithm $\Psi(\cdot)$ satisfies $\epsilon$-local differential privacy ($\epsilon$-LDP), where $\epsilon \ge 0 $, if and only if for any input $x_1$, $x_2$ and output $y$:
    \begin{equation}
        Pr[\Psi(x_1)=y]\leq e^\epsilon Pr[\Psi(x_2)=y].
    \end{equation}
    
\end{definition}
\noindent
Here, parameter $\epsilon$ (i.e., the \textit{privacy budget}) serves as a metric for quantifying the probability that an adversary can discern the input value based on the output. Thus, a smaller value of $\epsilon$ corresponds to a stronger privacy guarantee. 

\mypara{Frequency Oracle (FO)}
It is the most basic task in LDP, which entails estimating the frequency of a given value $x$ within a specific domain $\mathcal{D}$, and can serve as the building block for other more complex tasks.
In this paper, we adopt the \textit{optimized unary encoding} (OUE) mechanism as the FO protocol since it has optimal variance~\cite{usenix17_ldp}, leveraging the estimation outcome for subsequent analysis. 
Concretely, for each original value $x$ contributed by a user, it is first encoded into a length-$|\mathcal{D}|$ one-hot vector $V$, wherein only the $x$-th bit is set to 1. Subsequently, each user reports the perturbed encoding vector as follows:
\begin{equation}\label{equ:OUE-perturb}
        Pr[\hat{V}[i]=1] =
        \begin{cases}
            \frac{1}{2}, & \text{if } V[i] = 1\\
            \frac{1}{e^\epsilon + 1}, &\text{if } V[i] = 0,
        \end{cases}
    \end{equation}
where $\epsilon$ is the privacy budget and $\hat{V}$ is the reported noisy vector. On the curator side, the reported frequency $f'(x)$ is initially calculated by counting the vectors whose $x$-th bit is 1. Then, with $n$ denoting the number of participant users and $q=1/(e^\epsilon + 1)$, the frequency is adjusted as $\hat{f}(x)=(f'(x)/n-q)/(1/2-q)$. It has been proved~\cite{usenix17_ldp} that $\hat{f}(x)$ constitutes an unbiased estimation of the actual frequency with variance:
\begin{equation}\label{eq:var_oue}
    \var(\epsilon,n) = \frac{4e^\epsilon}{n(e^\epsilon-1)^2}
\end{equation}


\mypara{Composition}
To make the entire process adhere to the requirements of LDP, we rely on the fundamental properties of LDP~\cite{icmd18_cormode}:
\begin{theorem}[Sequential Composition]\label{theo:sequential_composition}
    Let $\Psi_1,\cdots,\Psi_k$ be a set of randomized mechanisms, where $\Psi_i$ satisfies $\epsilon_i$-LDP. Then, combining all the above subroutines with independent randomness results in a mechanism $\Psi$ satisfies $\sum_i^k\epsilon_i$-LDP.
\end{theorem}
\begin{theorem}[Post-Processing]\label{theo:post_processing}
    Post-processing the output of an LDP algorithm will not introduce additional privacy loss.
\end{theorem}

\subsection{LDP for Streaming Data}
In the stream setting, the emergence of data at each timestamp can be viewed as a sequence of \textit{events}. To address the privacy concerns in this dynamic context, the concepts of \textit{event-level} privacy and \textit{user-level} privacy are first proposed. 
Event-level privacy protects individual timestamps within a data stream, which may not provide sufficient privacy guarantees in realistic scenarios. 
In contrast, user-level privacy aims to conceal all timestamps of a data stream, offering stronger privacy assurances. However, it is not suitable for infinite streams, where it necessitates an infinite amount of perturbation. To strike a balance between event-level privacy and user-level privacy,  \textit{$w$-event privacy} is proposed to protect arbitrary $w$ consecutive timestamps in a stream. 

We begin by introducing several basic conceptions of data streams within the LDP setting. Let $T=\{c_1, c_2,\cdots\}$ represent a user's data stream, and $T[i]=c_i$ corresponds to the value at timestamp $i$. The \textit{stream prefix} $T_t$ is defined as $T_t=\{c_1,c_2,\cdots,c_t\}$, representing the sequence of values up to timestamp $t$. Building upon this, we can now elucidate the concept of $w$-neighboring:
\begin{definition}[$w$-neighboring~\cite{sigmod22_ldpids}]
    Two stream prefixes $T_t$, $T'_t$ are $w$-neighboring, if for each $T_t[i_1]$, $T_t[i_2]$, $T'_t[i_1]$, $T'_t[i_2]$ with $i_1\le i_2$, $T_t[i_1]\ne T'_t[i_1]$ and $T_t[i_2]\ne T'_t[i_2]$, it holds that $i_2-i_1+1\le w$.
\end{definition}
Two stream prefixes are $w$-neighboring means all their pairwise unequal values can fit in a window of up to $w$ timestamps.

\begin{definition}[$w$-event LDP~\cite{sigmod22_ldpids}]
    An algorithm $\Psi$ that takes a stream prefix $T_t=\{c_1,\cdots,c_t\}$ as input satisfies $w$-event $\epsilon$-LDP if for any $w$-neighboring stream prefixes $T_t$, $T'_t$ and all $t$, the output $y$ satisfies:
    \begin{equation}
        Pr[\Psi(T_t)=y]\le e^\epsilon Pr[\Psi(T'_t)=y].
    \end{equation}
\end{definition}
The $w$-event LDP can protect any sliding window of size $w$ for each user. When $w=1$, it will degenerate to event-level privacy; if $w$ is set as the length of a finite data stream, $w$-event privacy will converge towards user-level privacy.

\subsection{Problem Formulation}
Consider there is a number of mobile travelers who continuously report their locations to a curator for further analysis. 
Due to privacy concerns, the curator privately gathers the sensitive raw data at each timestamp and maintains a dynamically updated synthetic dataset. The synthetic dataset should retain a similar spatial-temporal distribution with the original trajectory streams and serve as a secure substitute for the original database. 

\begin{definition}[Private Trajectory Stream Synthesis]
Given the original database $\mathcal{T}_{orig}$ that consists of each user's trajectory stream $T_i^o=\{l_t|t=a_i,a_i+1, \cdots\}$, where $l_t=(x_t, y_t)$ denotes the two-dimensional coordinates of the $i$th user at timestamp $t$ and $a_i$ is the entering timestamp. The goal is to find a release algorithm $\Psi$, which takes inputs from $\mathcal{T}_{orig}$ and outputs a dynamic synthetic database $\mathcal{T}_{syn}$ at each $t$. $\Psi$ satisfies $w$-event $\epsilon$-LDP.
\end{definition}

\section{Our Proposal} 
\label{sec:method}

\begin{figure}[!t]
    \centering
    \includegraphics[width=0.48\textwidth]{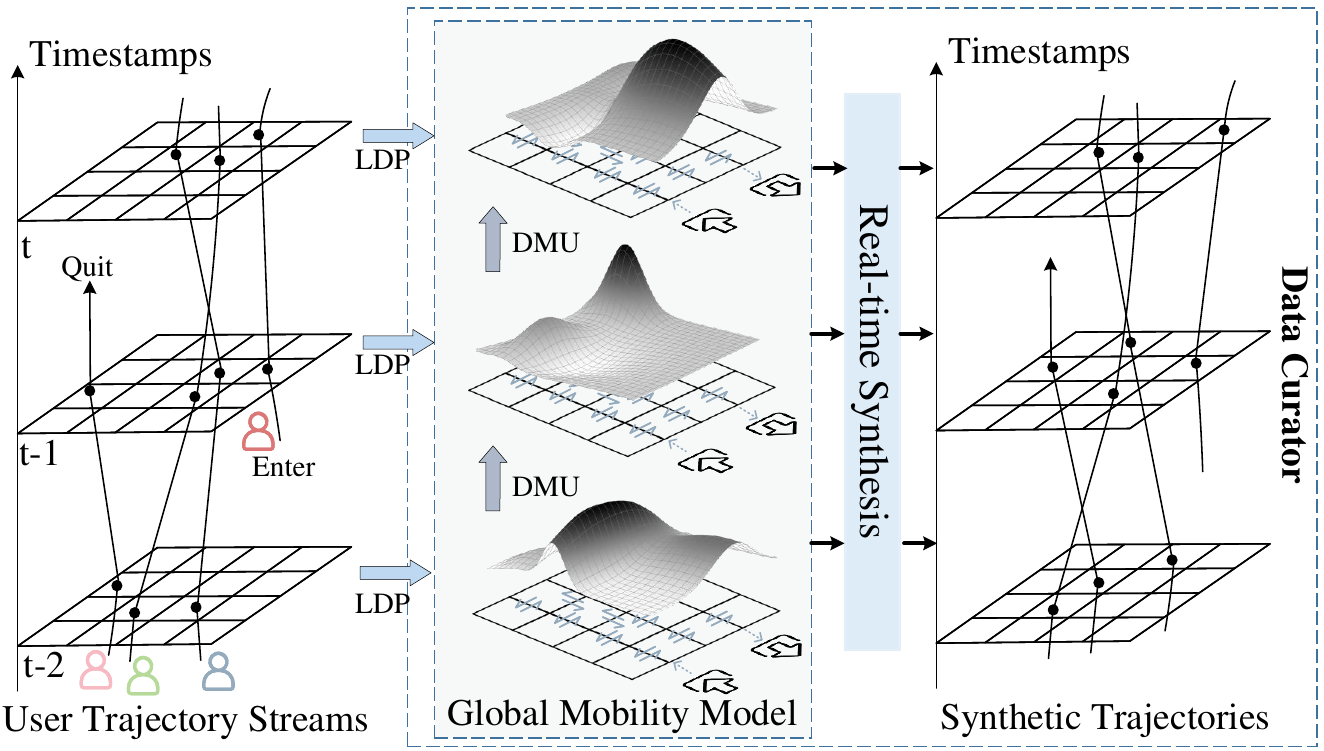}
    \caption{\mymethod architecture.}
    \label{fig:framework}
    \vspace{-4mm}
\end{figure}
In this section, we proceed to provide a method overview and the corresponding detailed techniques.

\subsection{Solution Overview}

\label{sec:overview}
We begin by introducing the overall architecture of \mymethod, as depicted in \autoref{fig:framework}. Users continuously share locations at discrete timestamps, forming trajectory streams. The curator collects them, maintaining an evolving synthetic dataset for publication. \mymethod mainly consists of four components: global mobility model, dynamic mobility update mechanism, real-time synthesis, and adaptive allocation strategy. 
The main procedure includes the following steps:

\begin{itemize}[leftmargin=*]
    \item \textbf{Step 1: Construction of Global Mobility Model (\autoref{sec:mobility_model}).} At each timestamp, each user first converts their spatial-temporal information into \textit{transition states} at the user side, which reflects her mobility status. The transitions are then collected by the curator through LDP protocol to protect sensitive information. Subsequently, they are aggregated at the curator side to construct the global mobility model, which extracts the mobility patterns of all reported users.
    \item \textbf{Step 2: Dynamic Mobility Update (\autoref{sec:dmu}).} Since the mobility patterns vary over time, we utilize a dynamic mobility update mechanism to update the global mobility model. To improve accuracy, the curator identifies the most informative transitions (\ie \textit{significant transitions}). 
    Afterwards, the distribution of significant transitions is updated using the reported perturbed statistics.
    \item \textbf{Step 3: Real-time Synthesis (\autoref{sec:synthesis}).} The final part is the real-time synthesis framework, which constructs a Markov chain-based model using the updated global mobility model and performs a generative process to update the synthetic database.  
\end{itemize}

\mypara{Adaptive Allocation (\autoref{sec:allocation})}
Besides, in streaming analysis, it is essential to appropriately distribute the privacy budget or users on timestamps, in order to satisfy the $w$-event LDP and preserve utility. To achieve this, \mymethod employs adaptive allocation strategies based on the dynamics of trajectory streams. We also track the status of users to maintain a dynamic active user set, which enables \mymethod to perform allocation strategies in more realistic scenarios.

\begin{figure}[t]
    \centering
    \includegraphics[width=0.48\textwidth]{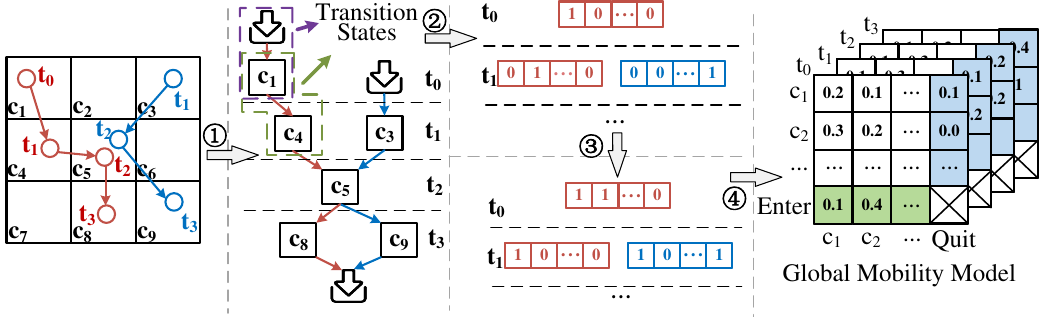}
    \caption{Illustration of mobility modeling. The process includes four steps: \ding{172} original streams are transformed into sequences of transition states; \ding{173} the transition states are encoded into binary vectors; \ding{174} LDP perturbation; \ding{175} curator side aggregation and model construction.}
    \label{fig:mobility_model}
    \vspace{-5mm}
\end{figure}

\subsection{Global Mobility Model}\label{sec:mobility_model}

To make the synthetic dataset retain a high resemblance to real trajectories, it is imperative to capture the mobility patterns inherent to the original trajectories. Traditional stream release methods, which collect isolated locations at individual timestamps, allow the curator merely to infer a static spatial distribution of users, and thus, fail to capture dynamic mobility patterns. To address it, we leverage the spatial-temporal context in data streams to extract the correlations in consecutive timestamps. In the streaming setting, our access is restricted only to data preceding the current timestamp. This suggests a preference for prior dependencies, where we assume the current location is determined by recent preceding locations. 

\mypara{Geospatial Discretization}
We initiate by discretizing the geospatial space. The raw spatial data is in a continuous two-dimensional domain, which makes pattern modeling complex. A widely accepted technique is to uniformly partition the entire space into $K\times K$ grid cells~\cite{icde13grid,vldb23_ldptrace}. Consequently, for each trajectory stream $T=\{(x_0,y_0), (x_1,y_1),\cdots\}$, it is transformed into a sequence of discrete grid cells: $T=\{c_0,c_1,\cdots\}$. We use $T[i]$ to denote the $i$th cell visited by $T$. 

\mypara{User Side Mobility Modeling}
At user side, the prior dependency at the current timestamp can be modeled using the prefix $\{T[t-r],\cdots, T[t]\}$, where $r$ is the dependency length. While incorporating longer dependencies might offer richer information, the state domain escalates exponentially in relation to $r$. Consequently, this would significantly increase the computational cost and elevate communication overhead. To ensure real-time analysis, it is appropriate to leverage shorter dependencies. Besides, as new locations in a stream appear one at a time, it is natural to use movement between two consecutive timestamps (\ie $r=1$) to represent the mobility status. Based on this, we define the \textit{movement transition state} $m_{ij}$ as the transition from cell $c_i$ to cell $c_j$.

Apart from standard movement transitions, our model integrates two special transitional events to better emulate the real-world dynamics, the \textit{entering} and \textit{quitting transition}. In realistic scenarios, users may not report their locations at every timestamp. For example, if a user enters regions with compromised signal reception (tunnels, mountainous areas, \etc) or opts to deactivate GPS service, their location becomes inaccessible to the curator, which causes a quitting event. Conversely, when a previously inactive user resumes location sharing or a new user arrives, an entering event emerges and initiates a new stream. These special transitions often reveal important spatial semantics of trajectories (\eg home/work places) and are very helpful in modeling the intrinsic spatial-temporal features. Besides, the distribution of entering/quitting events also plays an important role in our synthesis process, where we adjust the active synthetic trajectories dynamically to mimic the behaviors of real-world mobile travelers. Based on this, we define the entering and quitting transitions as follows:
\begin{definition}[Entering/Quitting Transitions]
    An entering transition $e_{i}$ represents the beginning of a new stream at cell $c_i$; a quitting transition $q_{j}$ indicates the cessation of a user's reporting activity, with the final reported location being $c_j$.
\end{definition}

We use $s_{ij}$ to represent general transition states that include the above three transitions (\ie $\mathcal{S}=\{s_{ij}\}=\{m_{ij}\}\bigcup \{e_i\}\bigcup \{q_j\}$). The discretized cell sequence of each user can equivalently be perceived as a sequence of transition states. At each timestamp, each user is in exactly one state encapsulating her mobility status. 
The process of user-side mobility modeling is depicted in \ding{172} of \autoref{fig:mobility_model}.

\mypara{Curator Side Mobility Modeling}
At curator side, we aim to maintain a global mobility model. Due to privacy concerns, the curator does not have direct access to the authentic data. To circumvent this, we employ LDP protocol (\eg OUE) to retrieve information from users privately. Specifically, each user's transition state $s_{u}$ is encoded into an $|\mathcal{S}|$-bit binary vector and then perturbed at user side (\ding{173} and \ding{174} in \autoref{fig:mobility_model}).
Subsequently, the collected statistics are aggregated to establish the global probability model predicated on a first-order Markov chain. In this model, the upcoming location's probability is solely determined by its immediate predecessor:
\begin{equation}
    Pr(T[i+1]=c|T[1]\cdots T[i])\\=Pr(T[i+1]=c|T[i])
\end{equation}
The comprehensive global mobility model consists of the movement distribution $\mathcal{M}$, entering distribution $\mathcal{E}$, and quitting distribution $\mathcal{Q}$. They can be found by aggregating the frequencies of all reported transition states. Specifically, let $f_{ij}$ denote the frequency of $m_{ij}$, then $\mathcal{M}$ can be calculated as $Pr(m_{ij})=\frac{f_{ij}}{\sum_{c_x\in\mathcal{C}}f_{ix}}$, where $\mathcal{C}$ is the domain of all grid cells. Note that the size of state space in $\mathcal{M}$ is $|\mathcal{C}|^2$, which may become too large when the discretization granularity $K$ increases. However, many transitions within this vast space will have a frequency of zero, implying they don't manifest in real-world scenarios.  Therefore, we only consider the transitions that satisfy reachability constraints. As an illustration, considering a 10-minute time granularity, it's unrealistic for a user in Beijing to traverse from the west 5th ring to the east 5th ring between two consecutive timestamps. In our uniformly split cells, we set the reachability constraints as transitions between adjacent cells, \ie $Pr(m_{ij})=\frac{f_{ij}}{\sum_{c_x\in\mathcal{N}_{c_i}}f_{ix}}$, where $\mathcal{N}_{c_i}$ is the set of neighboring cells of $c_i$, including itself. Following these constraints, our model only incorporates feasible transitions, thereby enhancing the realism and utility. The state space can be reduced to $\mathcal{O}(9|\mathcal{C}|)$. Similarly, $\mathcal{E}$ and $\mathcal{Q}$ can also be calculated by $f_{Ei}$ and $f_{jQ}$, which are the frequencies of $e_i$ and $q_j$, respectively. The final distribution is calculated as follows:
\begin{equation}\label{equ:mobility_distribution}
\begin{split}
    &Pr(m_{ij})=\frac{f_{ij}}{\sum_{c_x\in\mathcal{N}_{c_i}}f_{ix}+f_{iQ}}\\
    &Pr(e_{i})=\frac{f_{Ei}}{\sum_{c_x\in\mathcal{C}}f_{Ex}}, \quad Pr(q_{j})=\frac{f_{jQ}}{\sum_{c_x\in\mathcal{C}}f_{xQ}}
\end{split}    
\end{equation}
It's worth emphasizing that to make our synthesized trajectories authentically reflect real-world user dynamics, we consider the potential termination of a stream when synthesizing a trajectory. Therefore, we modify the original $Pr(m_{ij})$ by incorporating an additional term $f_{iQ}$ into the denominator, considering the frequency with which users quit at cell $c_i$.

Unlike existing solutions~\cite{sigmod22_ldpids}, which solely rely on static count statistics of users and overlook the correlation between consecutive timestamps, our global mobility model harnesses the spatial-temporal context effectively to capture the mobility patterns. Moreover, the incorporation of entering/quitting transitions also enhances the authenticity and utility of \mymethod.

\subsection{Dynamic Mobility Update (\dmu) Mechanism}\label{sec:dmu}
Using the frequency of transitions, \mymethod can effectively model the global mobility distribution. To dynamically update the mobility model at each timestamp, a straightforward approach is to directly substitute the extant statistics with the freshly gathered frequencies from users. 
This approach is effective when the available privacy budget is ample, leading to accurate frequency estimations. On the other hand, since many real-world streams demonstrate relatively consistent distributions between successive timestamps, it is plausible to assume that the extant mobility model still aligns closely with the genuine trajectories. Therefore, approximating current statistics with results from the most recent timestamp is also a feasible choice. Importantly, in realistic scenarios, the changing trends among different transition states may have great differences. For instance, during morning rush hours, main roads connecting residential areas to business districts might witness consistent traffic flows while transitions between other regions might experience considerable fluctuations. Consequently, at each timestamp, statistics of certain transitions can be approximated, whereas others require more precise updates. Therefore, a better strategy is to selectively update part of the mobility model and leave the remaining unchanged.

To achieve this, we propose to select the most informative parts of the transitions for update. Our primary objective is to pinpoint transitions that undergo substantial changes (termed as \textit{significant transitions}), making them difficult to approximate using the current mobility model. To guide the selection, we focus on the total introduced error caused by perturbation and approximation.
If a transition $s_{ij}$ is selected, \mymethod will use the perturbed statistics to update the mobility model. The introduced error can be calculated using the variance of OUE protocol: $Err_{upd}=\frac{4e^{\epsilon_t}}{n_t(e^{\epsilon_t}-1)^2}$, where $\epsilon_t$ is the privacy budget for perturbation and $n_t$ is the number of report users. For 
other transitions, their distribution in the mobility model will remain unchanged. Consequently, there will be distribution bias between the extant mobility model and the real one, which can be measured as $Err_{app}=|f_{ij}-\tilde{f}_{ij}|^2$, where $\tilde{f}_{ij}$ is the frequency collected from current global transition distribution. However, a challenge arises due to the unavailability of the real frequency $f_{ij}$ under LDP setting. To overcome this, we use the perturbed statistics $\hat{f}_{ij}$ to estimate the real frequency.

The selection is formulated as an optimization problem. For each transition state $s_{ij}$, it is associated with an indicator variable $x_{ij}$ that equals 1 if $s_{ij}$ is selected, and 0 otherwise. We aim to minimize the total introduced error of all transitions:
{\setlength{\abovedisplayskip}{3pt}
\setlength{\belowdisplayskip}{3pt}
\begin{equation}\label{equ:total_err}
        Err=\sum_{s_{ij}} x_{ij}\frac{4e^{\epsilon_t}}{n_t(e^{\epsilon_t}-1)^2}+ \sum_{s_{ij}}(1-x_{ij})|\tilde{f}_{ij}-\hat{f}_{ij}|^2
\end{equation}}When $\epsilon$ and $n_t$ are large, the variance of OUE will decrease, which means using the 
perturbed data can obtain more accurate results. 
On the other hand, the second term of $Err$ is data-dependent, which reflects how much the current global mobility model deviates from the real data distribution. If there's a notable change in the statistics of a transition state 
in consecutive timestamps, relying solely on the current approximation will result in substantial bias. When the potential bias is beyond the perturbation noise, \mymethod is likely to label the corresponding transition as the significant transition.

After obtaining the significant transitions $\mathcal{S}^*=\{s_{ij}|x_{ij}=1\}$, \mymethod use \autoref{equ:mobility_distribution} to update their distribution and the remaining transitions are unchanged.

The \dmu mechanism empowers \mymethod to achieve real-time updating and synthesis, a capability beyond the reach of other trajectory-aware solutions. Compared to existing real-time publish solution~\cite{sigmod22_ldpids}, we selectively update the most informative parts of the model, considering the varying trends between transitions. This approach yields a substantial reduction in the total introduced error, thereby enhancing utility.

\subsection{Real-time Trajectory Synthesis}\label{sec:synthesis}
\mymethod builds a probabilistic model for private synopsis according to the global mobility model. The algorithm consists of two steps: new point generation and size adjustment. 

\mypara{New Point Generation} 
For every extant trajectory stream within the current synthetic dataset, a new location cell is appended according to the global mobility model. Specifically, we leverage the Markov chain probability to ensure that the movement of synthetic trajectories aligns with the authentic distribution: $ Pr(c_{next}=c_j|c_{current}=c_i)=Pr(m_{ij})$.
To improve the authenticity, we also consider the potential termination of the current synthetic trajectory: $Pr(quit|c_{current}=c_i)=\frac{f_{iQ}}{\sum_{c_x\in\mathcal{N}_{c_i}}f_{ix}+f_{iQ}}$.
While the first-order Markov model predominantly concentrates on single-step transitions, direct application of this probability might inadvertently lead to premature stream termination. To further augment the model's authenticity and informativeness, we incorporate the current stream's length $\ell$ and reweight the quitting probability:
\begin{equation}\label{equ:quit}
    Pr(quit|c_{current}=c_i)=\frac{\ell}{\lambda}\cdot \frac{f_{iQ}}{\sum_{c_x\in\mathcal{N}_{c_i}}f_{ix}+f_{iQ}}
\end{equation}
where $\lambda$ is a factor controlling the effect of stream length. 

\mypara{Size Adjustment}  
In realistic scenarios, the entering and quitting events lead to fluctuations in the number of active users. In many data analysis tasks, it is crucial to query an accurate number of users or spatial points, such as traffic congestion control and emergency response. Therefore,  it's important to ensure that the size of $\mathcal{T}_{syn}$ mirrors the size of $\mathcal{T}_{orig}$ at each timestamp. To maintain congruity in the spatial distribution of entering/quitting users, we utilize the entering and quitting distributions ($\mathcal{E}$ and $\mathcal{Q}$). When the number of real users is greater than that in $\mathcal{T}_{syn}$, \mymethod appends new trajectories by sampling the start cell from $\mathcal{E}$:
    $Pr(c_{start}=c_i)=Pr(e_{i})$.
This strategy is also invoked during the initialization of $\mathcal{T}_{syn}$. Similarly, when the size of $\mathcal{T}_{syn}$ surpasses $\mathcal{T}_{orig}$, \mymethod will terminate a subset of existing streams according to their most recent locations: $Pr(quit|c_{last}=c_j)=Pr(q_j)$.
\subsection{Adaptive Allocation Strategy}
\label{sec:allocation}

$w$-event LDP requires that the cumulative budgets within any sliding window of size $w$ must not exceed $\epsilon$. This can be achieved through either budget- or population-division strategies. While previous work~\cite{sigmod22_ldpids} introduces several adaptive allocation approaches, 
they either necessitate a consistent user number or require the curator to pre-group fixed users into exclusive subsets. This makes them inapplicable in a more realistic setting, where active users vary over time. To address this, we propose a portion-based mechanism that can be implemented in both budget and population division strategies.

\mypara{Budget-division Strategy}
Based on \autoref{theo:sequential_composition}, we can utilize the privacy budget by appropriately distributing $\epsilon$ on individual timestamps.  Specifically, the holistic processing mechanism $\Psi$ can be seen as a sequential composition of sub-mechanisms $\Psi_1$, $\Psi_2$,$\cdots$. For every timestamp $i$, $\Psi_i$ collects users' data and updates the mobility model with a budget of $\epsilon_i$. In our portion-based approach, at each timestamp, the curator first calculates the remaining budget at the current window: $\epsilon_{rm}=\epsilon-\sum_{i=t-w+1}^{t-1} \epsilon_i$.  Afterwards, it allocates a portion $p$ of $\epsilon_{rm}$ for perturbation to ensure that the total consumed budget in the sliding window remains within the bound of $\epsilon$.


\mypara{Population-division Strategy}
Based on \autoref{eq:var_oue}, it is evident that 
the perturbation variance is less sensitive to $n$ than $\epsilon$.
Therefore, recent studies~\cite{sigmod22_ldpids,tdsc19_tianhao} are exploring population-division methods where users rather than the privacy budget are partitioned. For each report, a group of users is selected and leverages the entire $\epsilon$ for perturbation.
The challenge lies in achieving population allocation on an evolving user set.
To address it, we track the status of users and maintain a dynamic active user set, which will be detailed in \autoref{sec:algorithm}. For each report,  we first determine a portion $p$ and the final allocated population size becomes $p\cdot |U_a|$, where $U_a$ is the active user set. Once $p$ is decided, the curator randomly selects $p$ portion of the active users and collects their data via LDP protocol.

\mypara{Determination of $p$}
Intuitively, the allocated portion is relative to the changing trends of the data. To model the dynamics of data streams, we define the \textit{deviation} as:
\begin{equation}
    Dev_t=\sum_{s_{ij}\in\mathcal{S}}(f_{ij}^{t-1}-\frac{1}{\kappa}\sum\nolimits_{k=t-\kappa-1}^{t-1}f_{ij}^{k})
\end{equation}
where $f_{ij}^t$ is the frequency of transition state $s_{ij}$ at timestamp $t$. $Dev$ mirrors the magnitude by which the most recent statistics differ from prior ones. 
An increase in $Dev$ suggests a stream being less uniform, thereby potentially increasing approximation bias. To ensure the update accuracy, a larger $p$ is advisable. We employ the logarithm function to represent their positive correlation since it grows slowly when $Dev$ is large, aiding the curator in avoiding excessive use of privacy budgets/users when sudden large changes occur in the stream. On the other hand, rapid data changes may suggest a rise in the total number of significant transitions $|\mathcal{S}^*_t|$ to be updated within the current window. 
Thus, we monitor the changing speed by considering the ratio of $|\mathcal{S}^*_t|$ to $|\mathcal{S}|$
To prevent premature exhaustion of budget or users, a smaller $p$ is allocated when the ratio increases.
Besides, under the same conditions, a larger window size $w$ suggests more timestamps to be protected, thus the allocated $p$ on the current timestamp should be reduced. The final allocation portion $p$ can be calculated as follows:

\begin{equation}
    p_t=\min\{\frac{\alpha}{w}(1-\frac{1}{\kappa}\sum_{i=t-\kappa-1}^{t-1}\frac{|\mathcal{S}^*_i|}{|\mathcal{S}|})\ln (Dev_t+1), p_{max}\}
    \label{equ:allocation_p}
\end{equation}
where $\alpha$ and $\kappa$ are hyperparameters to control the scale of $p$ and the number of recent timestamps taken into consideration; $p_{max}$ represents the maximum portion constraint and is set to 0.6 in our experiments. This constraint can prevent excessive usage of the budget/users within a single timestamp. 

Besides the data-dependent approaches, we also implement two straightforward methods: Uniform and Sample. For Uniform in budget-division strategy, we evenly distribute the entire budget on each timestamp (\ie $\epsilon_i=\epsilon/w$). For the population-division strategy, we assign $p=1/w$ for each timestamp. In the Sample approach, the entire budget is dedicated to the first timestamp of each window. This means that all active users report every $w$ timestamps with a budget of $\epsilon$. For other timestamps, $p$ remains zero and the global mobility model is not updated.
There is also an alternative population-division strategy where users randomly select a timestamp in the current window to report after their entrance. We have tested its performance and found that our adaptive methods generally outperform it, especially in skewed and complex datasets. We also notice that in certain cases the random strategy exhibits superiority due to its less user wastage.

\begin{algorithm}[t]
\small
\caption{\mymethod with Population-Division}
\label{alg:population}
\LinesNumbered
\linespread{0.8}\selectfont
\KwIn{Raw stream dataset $\mathcal{T}_{orig}$, budget $\epsilon$, window size $w$}
\KwOut{Synthetic trajectory stream database $\mathcal{T}_\text{syn}$}
    Set the status of each new user as $active$\; \label{alg:1}
    Randomly sample $1/w$ of the users, denoted as $U_1$\; \label{alg:2}
    \For{each user $u\in{U_1}$}{ \label{alg:3}
        $V_u \gets$ OUE($s_u$, $\epsilon$);
        $u.status=inactive$\;
    }
    Initialize the global mobility model and $\mathcal{T}_{syn}$\;\label{alg:5}
    \For{each timestamp $t\ge 2$}{
        Set the status of each newly arrived user as $active$\;\label{alg:7}
        Set the status of each quitted user as $quitted$\;\label{alg:8}
        Recycle users on timestamp $t-w$\;\label{alg:9}
        
        Calculate $p_t$ using \autoref{equ:allocation_p}\;\label{alg:10}

        $U_A=\{u|u.status=active\}$\; 
        Randomly sample $U_t$ from $U_A$ with the size of $p_t\cdot |U_A|$\; \label{alg:12}
        \For{each user $u\in U_t$}{\label{alg:13}
            $V_u \gets$ OUE($s_u$, $\epsilon$);
            $u.status=inactive$\;\label{alg:14}
        }
        Select $\mathcal{S}^*$ by minimizing \autoref{equ:total_err} \;\label{alg:15}
        Update global mobility model using $\mathcal{S}^*$\;\label{alg:16}
        Generate a new point for each trajectory $T\in \mathcal{T}_{syn}$\;\label{alg:17}
        Adjust the size of $\mathcal{T}_{syn}$ based on number of active users\;\label{alg:18}
        
    }
\end{algorithm}

\subsection{Putting Things Together: \mymethod}
\label{sec:algorithm}

We proceed to describe the overall workflow of \mymethod. 
Due to space limits, we take the population-division strategy as an example, as illustrated in \autoref{alg:population}.  
At the first timestamp, we allocate $1/w$ of users to initialize the mobility model and $\mathcal{T}_{syn}$ (\autoref{alg:1}-\autoref{alg:5}). For the incoming timestamps, the curator first registers the new-come users (\autoref{alg:7}) and removes the users who cease location sharing (\autoref{alg:8}).  After determining the allocation portion $p$, we randomly sample $p$ portion of the active users for reporting (\autoref{alg:10}-\autoref{alg:12}). These chosen users subsequently report their data with LDP protocol, after which they are designated as \textit{inactive} (\autoref{alg:13}-\autoref{alg:14}). The curator then performs \dmu to update the global mobility model (\autoref{alg:15}-\autoref{alg:16}) and finally employs the real-time synthesis to generate current $\mathcal{T}_{syn}$ (\autoref{alg:17}-\autoref{alg:18}).
Since $w$-event LDP protects the privacy of any window of size $w$, the reported users that lie outside the current window should be recycled. Hence, before the collection, the curator reviews the reported users at timestamp $t-w$. Users whose status is \textit{inactive} rather than \textit{quitted} will be reset to \textit{active} and be ready for the next report (\autoref{alg:9}).

\section{Theoretical Analysis}\label{sec:analysis}
In this section, we provide privacy and complexity analysis.

\subsection{Privacy Analysis}\label{sec:privacy_analysis}
 We first prove the privacy guarantee of \mymethod.
\begin{theorem}\label{theo:LDP}
    \mymethod satisfies $w$-event $\epsilon$-LDP for each user.
\end{theorem}
\begin{proof}
    Please find detailed proof in our technical report~\cite{fullversion}.


\end{proof}

\subsection{Complexity Analysis}
We also discuss the computational and communication cost of \mymethod at each timestamp. 

\mypara{User Side Computation}
Since users retain control of their own data, the perturbation is performed locally on the user side. For each user, each bit of the encoded vector is independently perturbed. Therefore, the cost of reporting operation per user is $\mathcal{O}(|\mathcal{S}|)$, where $\mathcal{S}$ is the domain of transition states. Since we only consider the neighboring transitions, the complexity is $\mathcal{O}(9|\mathcal{C}|)$, where $\mathcal{C}$ is the domain of all cells.

\mypara{Curator Side Computation}
The curator first determines the allocated portion $p$ based on the deviation and number of significant transitions, which exhibits a complexity of $\mathcal{O}(|\mathcal{S}|)$ $(i.e., \mathcal{O}(9|\mathcal{C}|))$. After collecting the transition states, the unbiased adjustment of OUE protocol exhibits a complexity of $\mathcal{O}(n\cdot|\mathcal{C}|)$, where $n$ is the number of report users at the current timestamp. To perform the \dmu mechanism, the optimization of \autoref{equ:total_err} and the update process can both be done in $\mathcal{O}(9|\mathcal{C}|)$ complexity. In the synthesis phase, new point generation incurs a cost of $\mathcal{O}(n_l)$ and the size adjustment takes a complexity of $\mathcal{O}(|n_l-n|)$, where $n_l$ is the number of active users at the last timestamp. Consequently, the overall overhead for each timestamp is $\mathcal{O}(n|\mathcal{\mathcal{C}}|+n_l)$. For population division-based methods, the sampling of report users takes an additional complexity of $\mathcal{O}(p\cdot|U_a|)$.

\mypara{Communication Overhead}
Now we analyze the communication cost of \mymethod. At each timestamp, the communication between the curator and distributed users is mainly the report process. For each report, every user transmits the perturbed encoding vector to the curator, thus the communication overhead for this process is the length of the encoding vector. For budget-division based methods, all users will participate in the reporting process, thus the total communication bits per timestamp is $\mathcal{O}(n\cdot9|\mathcal{C}|)$. For population-division based methods, since only a portion of sampled users communicates with the curator, the overhead is $\mathcal{O}(p\cdot n\cdot 9|\mathcal{C}|)$. Moreover, as the curator needs to track the status of each user and correspondingly inform users whether to report their data, there is an extra communication overhead of $\mathcal{O}(n)$.

\section{Experiments}\label{sec:experiments}
In this section, we evaluate the performance of the proposed method compared with existing state-of-the-arts.
\subsection{Experimental Setup}
\mypara{Datasets}
Three trajectory datasets are used: T-Drive, Oldenburg and SanJoaquin. T-Drive~\cite{tdrive}
records the traces of 10,357 taxis operating in Beijing during one week. We select the denser area within the 5th ring and follow~\cite{sigmod22_ldpids} to transform the time dimension into 886 timestamps with a granularity of 10 minutes. Oldenburg and SanJoaquin are generated using Brinkhoff's network generator for moving objects  \cite{brinkhoff}. Specifically, we use the roadmap of Oldenburg city to create the Oldenburg dataset with 500 timestamps. There are 10,000 users at the beginning and 500 new users are added per timestamp. SanJoaquin is based on the map of San Joaquin County and contains 1,000 timestamps. It begins with 10,000 initial users,  and an additional 1,000 users are included per timestamp. The users in these two datasets randomly quit sharing their locations and the time interval between two consecutive timestamps approximates 15 seconds. We assume the curator periodically collects the locations from users, and align the time in three datasets with corresponding discrete collection timestamps. For trajectories including non-adjacent timestamps, we add quitting events and split them into multiple streams. The detailed dataset statistics are illustrated in \autoref{tab:dataset_description}.



\begin{table}[!t]
    \centering
    \vspace{-8mm}
    \caption{Statistics of the datasets used in our experiments.}
    \vspace{-9pt}
    \resizebox{0.48\textwidth}{!}{
    \begin{tabular}{ccccc}
    \toprule
         \textbf{Dataset} &\textbf{Size} &\textbf{\# of Points} &\textbf{Average Length} &\textbf{Timestamps} \\ \midrule
         \textbf{T-Drive} &232,640 &3,167,316 &13.61 &886 \\ 
         \textbf{Oldenburg} &260,000 &15,597,242 &59.98 &500\\ 
         \textbf{SanJoaquin} &1,010,000 &55,854,936 &55.30 &1,000 \\ \bottomrule
    \end{tabular}
    }
    \vspace{-6mm}
    \label{tab:dataset_description}
\end{table}


\mypara{Baselines}
To the best of our knowledge, \ldpids~\cite{sigmod22_ldpids} is the state-of-the-art streaming release framework that satisfies rigorous $w$-event $\epsilon$-LDP. It proposes four strategies:
\begin{itemize}[leftmargin=*]
    \item \textbf{LBD} and \textbf{LBA} are two budget division based methods. LBD distributes the budget to the perturbation timestamps (\ie sampling points) in an exponentially decreasing manner. LBA uniformly allocates the privacy budget and the budget of non-sampling points are absorbed by sampling points.
    \item \textbf{LPD} and \textbf{LPA} are two population division based methods. The allocation strategy is similar to LBD and LBA. The difference is to separate users rather than the privacy budget.
\end{itemize}
However, since \ldpids is designed to solve the histogram release problem, it is not directly applicable for publishing trajectories. Therefore, we make some modifications for a fair comparison. Specifically, we employ its two-step private mechanism to collect the transition states from users and build the global mobility model. Afterward, we leverage the same Markov probability model as ours to generate new points without considering the entering/quitting of users. 

We use \mymethodb and \mymethodu to denote our proposed budget and population division strategies.

\mypara{Experimental Settings}
In our experiments, we study the impact of parameters, as summarized in \autoref{tab:parameters}. The parameters $\epsilon$, $w$, and $\varphi$ (introduced in \autoref{sec:metrics}) only have an impact on utility performance; we analyze them in \autoref{sec:ablation}. In \autoref{sec:scalability}, we analyze the efficiency and scalability of \mymethod \wrt data cardinality and $K$.  For other parameters that have relatively little impact, we fix their values.
Specifically, the termination restriction factor $\lambda$  is set as the average trajectory length of each dataset, since it reflects the overall propensity of users to quit. For the adaptive allocation strategy, we set $\alpha=8$ and $\kappa=5$. Our experiments are conducted on a computer with Intel Xeon 2.1GHz CPU and 128 GB main memory.

\begin{table}[!t]
    \centering
    \vspace{-8mm}
    \caption{Parameter ranges. The default values are in bold.}
    \vspace{-9pt}
    \resizebox{0.42\textwidth}{!}{
    \begin{tabular}{lc}
    \toprule
         \textbf{Parameter} &\textbf{Range}  \\ \midrule
         privacy budget $\epsilon$ &0.5, \textbf{1.0}, 1.5, 2.0 \\ 
         window size $w$ &10, \textbf{20}, 30, 40, 50\\ 
         evaluation time range size $\varphi$ &5, 10, \textbf{20}, 50, 100 \\
         discretization granularity $K$ &2, \textbf{6}, 10, 14, 18 \\
         size of datasets &20\%, 40\%, 60\%, 80\%, \textbf{100\%} \\ \bottomrule
    \end{tabular}
    }
    \vspace{-6mm}
    \label{tab:parameters}
\end{table}

\begin{table*}[t]
\centering
\vspace{-8mm}
\caption{Overall utility performance with different privacy budgets. The best values are shown in bold. For Hotspot NDCG, Pattern F1 and Kendall Tau, larger values are better. For other metrics, smaller values are better.} \vspace{-3mm}
\label{tab:overall-performance}
\resizebox{0.98\textwidth}{!}{
\begin{tabular}{c|c|c c c c| c c c c| c c c c}
\toprule
\multicolumn{2}{c|}{} &
\multicolumn{4}{c|}{\textbf{T-Drive}} &
\multicolumn{4}{c|}{\textbf{Oldenburg}} &
\multicolumn{4}{c}{\textbf{SanJoaquin}}\\
\multicolumn{2}{c|}{} &
\multicolumn{1}{c}{$\epsilon=0.5$} &
\multicolumn{1}{c}{$\epsilon=1.0$} &
\multicolumn{1}{c}{$\epsilon=1.5$} &
\multicolumn{1}{c|}{$\epsilon=2.0$} &
\multicolumn{1}{c}{$\epsilon=0.5$} &
\multicolumn{1}{c}{$\epsilon=1.0$} &
\multicolumn{1}{c}{$\epsilon=1.5$} &
\multicolumn{1}{c|}{$\epsilon=2.0$} &
\multicolumn{1}{c}{$\epsilon=0.5$} &
\multicolumn{1}{c}{$\epsilon=1.0$} &
\multicolumn{1}{c}{$\epsilon=1.5$} &
\multicolumn{1}{c}{$\epsilon=2.0$}
\\ \hline
\multirow{6}*{\textbf{Density Error}}
&LBD     &0.5162 &0.6028 & 0.6090&0.6248    &0.2342&0.2262& 0.2591
&0.4859&0.6456 &0.4957 & 0.2261
&0.6039\\
&LBA     &0.6418 &0.6415 & 0.4867&0.5005    &0.2542&0.3245 & 0.2867
&0.1959    &0.5674 &0.4513 & 0.6470
&0.5660\\
&LPD     &0.3738 &0.2173 & 0.5640&0.5979    &0.5166 &0.5161 & 0.2380
&0.4121    &0.6226 &0.2910 & 0.5216
&\B{0.0735}\\
&LPA     &0.2617&0.5707 & 0.3126&0.2295    &0.6049 &0.2334 & 0.3799
&0.3606    &0.2833 &0.5189 & 0.3217
&0.3962\\
&\mymethodb  & 0.1398& 0.1354& 0.1358&0.1342     & 0.1321& 0.1260& 0.1267
&0.1242   &0.1696 &0.1636 & 0.1569
&0.1543\\
&\mymethodu  & \B{0.1365}& \B{0.1338}& \B{0.1319}&\B{0.1271} & \B{0.1259}& \B{0.1171}& \B{0.1115}&         \B{0.1033}&\B{0.1549} &\B{0.1435} &\B{0.1299} &0.1164\\ 
 \hline
\multirow{6}*{\textbf{Query Error}}
&LBD     & 1.6495
& 1.8010
&     1.8129
& 1.8308
& 0.7628&     0.7559& 0.7318
& 0.9204&0.9310   &0.8909 &0.8673 &0.9510\\
&LBA     & 1.8691
& 1.8678
&     1.5565
& 1.6070
& 0.8512
&     0.9101
& 0.9705
& 0.7862
&0.8834     &0.8753 &0.9457 &0.8832\\
&LPD     & 0.9960
& 0.7318
&     1.3166
& 1.0260
& 1.5086
&     0.7061
& 0.7591
& 0.7845
&     0.9522&0.8838 &0.9341 &0.8223\\
&LPA     & 0.9449& 1.3170&     0.8615
& 1.1708
& 0.9600
&     0.8472
& 0.8062
& 0.9306
&0.8368     &0.9104 & 0.8637&0.8618\\
&\mymethodb     & 0.5124 & 0.5098&     0.5075
& 0.4997
& 0.6084
&     0.5958
& 0.5574
& 0.6200
&0.5371     &0.5467 & 0.5280&0.4944\\
&\mymethodu  & \B{0.5055}
& \B{0.4851}&     \B{0.4764}& \B{0.4536}& \B{0.5960}
&     \B{0.5629}
& \B{0.5351}
& \B{0.5560}
&\B{0.4988}     &\B{0.5110} &\B{0.4790} &\B{0.4599}\\ 
 \hline
\multirow{6}*{\textbf{Hotspot NDCG}}
&LBD     & 0.1539
& 0.1970
&     0.2805
& 0.0478
& 0.1255&     0.0948& 0.0978& 0.1010 &0.1248 &0.2558 &0.5204 &0.3052\\
&LBA     & 0.0831
& 0.0912
&     0.1906
&0.2325 &0.1339 &0.1460     &0.3129 &0.3390 &0.1230  &0.4219 &0.1386 &0.1222\\
&LPD     & 0.1859
& 0.3703
&     0.1872
& 0.1475&0.1487 &0.0939     &0.3189 &0.1141 &0.0924  &0.3750 &0.2127 &0.6016\\
&LPA     & 0.0354& 0.2538&     0.1683
& 
0.1946&0.1285 &0.1294     &0.3414 &0.0571 &0.0277     &0.2064 &0.0522 &0.3650\\
&\mymethodb     & 0.3917
& \B{0.4824}
&     0.4595
& 0.4516&0.4381 &0.4613     &0.4964 &0.4676 &0.4751    &0.5641 & 0.6662&0.7158\\
&\mymethodu  & \B{0.4428}
& 0.4697
&     \B{0.5380}
& \B{0.5416} &\B{0.4725} &\B{0.5908} &\B{0.6821} &\B{0.6988} &\B{0.6803}     &\B{0.7913} &\B{0.8880} &\B{0.9129}\\ 
 \hline\hline
\multirow{6}*{\textbf{Transition Error}}
&LBD     &0.6094&0.6386& 0.6341&0.6551&0.5697&0.4955& 0.5003
&0.6084   &0.6736 &0.5799 & 0.4589
&0.6438\\
&LBA     &0.6575 &0.6573 & 0.5807
&
0.5848&0.5892 &0.5836 & 0.5591
&0.5274   &0.6450 &0.5685 & 0.6771
&0.6422\\
&LPD     &0.5127 &0.4509 & 0.6205
&0.6143&0.6073 &0.5336 & 0.4291
&0.4740   &0.6570 &0.4610 & 0.6095
&\B{0.2973}\\
&LPA     &0.5430&0.6108 & 0.5047
&
0.4631&0.6632 &0.4711 & 0.5289
&0.5089  &0.5710 &0.6178 & 0.5559
&0.4609\\
&\mymethodb  & 0.4171& 0.4104& 0.4091
&0.4063& 0.4906& 0.4745& 0.4669
&0.4539   &0.5021 &0.4827 & 0.4682
&0.4517\\
&\mymethodu  & \B{0.4078}& \B{0.3951}& \B{0.3798}
&\B{0.3724}& \B{0.4645}& \B{0.4223}& \B{0.3907}
&         \B{0.3581}&\B{0.4617} &\B{0.4134} &\B{0.3740} &0.3347\\ 
 \hline
\multirow{6}*{\textbf{Pattern F1}}
&LBD     & 0.2029
& 0.1760
&     0.1873
& 0.0912
& 0.2633&     0.2394& 0.1950
& 0.2418&0.1125     &0.2075 &0.2756 &0.1784\\
&LBA     & 0.0955
& 0.0973
&     0.2470
& 0.2291
& 0.1468
&     0.1515
& 0.2089
& 0.2247
&0.2430     &0.2400 &0.2496 &0.2423\\
&LPD     & 0.2487
& 0.2715
&     0.1586
& 0.1272
& 0.1804
&     0.1203
& 0.3007
& 0.3052
&     0.1719&0.2007 &0.2668 &0.4055\\
&LPA     & 0.1539& 0.289&     0.2763
& 0.2467
& 0.1207
&     0.2289
& 0.2795
& 0.2666
&0.2682     &0.2656 & 0.2603&0.3527\\
&\mymethodb     & 0.3668
& 0.3898
&     0.3876
& 0.4004
& 0.3988
&     0.4185& 0.4311
& 0.4296
&0.4067     &0.4304 & 0.4381&0.4421\\
&\mymethodu  & \B{0.3958}& \B{0.4093}
&     \B{0.4128}
& \B{0.4312}
& \B{0.4249}
&     \textbf{0.4596}& \B{0.4561}
& \B{0.4768}
&\B{0.4407}     &\B{0.4600} &\B{0.4516} &\B{0.4808}\\ 
 \hline\hline
\multirow{6}*{\textbf{Kendall Tau}}
&LBD     &0.1651&0.2413& 0.1143&-0.0032
&0.3635&0.2460& 0.2429&0.1159&0.2429 &0.2175 & 0.2714
&0.1254\\
&LBA     &0.1809&0.1764 & 0.0889
&0.1175   &0.1793 &0.2429 & 0.1254
&0.2111   &0.1794 &0.2905 & 0.2429
&0.1794\\
&LPD     &0.2476 &0.4095 & 0.0317
&0.2476   &0.1825 &0.1952 & 0.1000
&0.1603   &0.1794 &0.0048 & 0.3635
&0.3317\\
&LPA     &-0.0698&0.3683 & 0.1143
&-0.0070   &0.2746 &0.2460 & 0.4904
&0.2587   &0.0143 &0.2778 & 0.1730
&0.3381\\
&\mymethodb  & 0.6730& 0.6952& \B{0.7143}
&0.7270   & 0.7285& 0.7540& 0.7095
&\B{0.7413}&\B{0.6651} &0.6524 & 0.6460
&0.6524\\
&\mymethodu  & \B{0.7333}& \B{0.6984}& 0.7079
& \B{0.7429}& \B{0.7317}&     \B{0.7635}& \B{0.7222}&   0.7000&0.6587 &\B{0.6778} &\B{0.6730} &\B{0.6809}\\ 
 \hline
\multirow{6}*{\textbf{Trip Error}}
&LBD     &0.6779 &0.6834 & 0.6834
&0.6834   &0.4217&0.3736& 0.4052
&0.5833&0.6751 &0.4711 & 0.4486
&0.6852\\
&LBA     &0.6834 &0.6834 & 0.6657
&0.6737   &0.3905 &0.4319 & 0.4259
&0.3499   &0.6798 &0.6482 & 0.6768
&0.6799\\
&LPD     &0.5891 &0.4731 & 0.6834
&0.6432   &0.6579 &0.5664 & 0.3944
&0.5246   &0.6823 &0.4398 & 0.6657
&\B{0.3260}\\
&LPA     &0.4614&0.6676 & 0.5055
&0.4524   &0.6719 &0.3677 & 0.5916
&0.5369   &0.4395 &0.6256 & 0.4736
&0.6119\\
&\mymethodb  & 0.3390& 0.3362& 0.3290
&0.3346 & 0.2979& 0.2935& 0.2892
&   0.2853&0.3650 &0.3602 & 0.3574
&0.3623\\
&\mymethodu  &\B{0.3275} & \B{0.3227}& \B{0.3100}
&       \B{0.3055}& \B{0.2961}& \B{0.2860}& \B{0.2824}
&   \B{0.2756}&\B{0.3543}  &\B{0.3499} &\B{0.3397} &0.3353\\ 
 \hline
 \multirow{6}*{\textbf{Length Error}}
&LBD     & 0.6931& 0.6931&     0.6931& 0.6931& 0.6931&     0.6931& 0.6931& 0.6931&     0.6931& 0.6931& 0.6931&0.6931
\\
&LBA     & 0.6931& 0.6931&     0.6931& 0.6931& 0.6931&     0.6931& 0.6931& 0.6931&     0.6931& 0.6931& 0.6931&0.6931
\\
&LPD     & 0.6931& 0.6931&     0.6931& 0.6931& 0.6931&     0.6931& 0.6931& 0.6931&     0.6931& 0.6931& 0.6931&0.6931
\\
&LPA     & 0.6931& 0.6931&     0.6931& 0.6931& 0.6931&     0.6931& 0.6931& 0.6931&     0.6931& 0.6931& 0.6931&0.6931
\\
&\mymethodb     & 0.2016& 0.2013&     0.1966& 0.1883
&0.5350 &\B{0.5197} &0.5108 &0.5093   &0.4964 &0.5063 &0.4838 &0.5078\\
&\mymethodu  & \B{0.2020}& \B{0.1915}&  \B{0.1828} & \B{0.1754}   &\B{0.5168} &0.5230 & \B{0.5107}&\B{0.5092}  &\B{0.4858} &\B{0.4857} &\B{0.4855} &\B{0.4430}\\ 
 \bottomrule

\end{tabular}}
\vspace{-15pt}
\end{table*}

\subsection{Utility Metrics}\label{sec:metrics}
To evaluate the effectiveness and versatility of \mymethod, we implement various utility metrics, encompassing both real-time analysis and static historical analysis.

\mypara{Streaming Metrics}
In streaming analysis, we focus on the performance of an individual timestamp or a period within the stream. We categorize the metrics into two aspects: global and semantic levels.
The global level utility assesses the overarching spatial-temporal distribution of the dataset:
\begin{itemize}[topsep=0pt,itemsep=0pt,parsep=0pt,partopsep=0pt,leftmargin=*]
    \item \textbf{Density Error} measures the difference (Jenson-Shannon divergence JSD~\cite{ccs18_adatrace,vldb23_ldptrace}) between the density distribution of  $\mathcal{T}_{syn}$ and  $\mathcal{T}_{orig}$ in a given timestamp.
    \item \textbf{Query Error} is widely used in evaluation of synthesis-based algorithms~\cite{ccs18_adatrace,vldb23_ldptrace}. We use spatial-temporal range query to measure the utility within a time period of size $\varphi$. Specifically, a query $Q_i(\mathcal{T})$ returns the count of spatial points in dataset $\mathcal{T}$ that fall within a specific spatial region during a time range of size $\varphi$.
    The query error can be calculated as the mean relative error between $\mathcal{T}_{orig}$ and $\mathcal{T}_{syn}$.
    We report the average result of 100 random queries.
    \item \textbf{Hotspot NDCG} measures how well the synthetic dataset preserves the spatial-temporal hotspots. We use Normalized Discounted Cumulative Gain (NDCG@$n_h$) to evaluate the quality of ranking the most popular $n_h$ cells within a random time range. We set $n_h=10$ and report the average result of 100 random time ranges of size $\varphi$.
\end{itemize}

The semantic level utility measures the preservation of mobility patterns in the original trajectories. 
\begin{itemize}[topsep=0pt,itemsep=0pt,parsep=0pt,partopsep=0pt,leftmargin=*]
    \item \textbf{Transition Error} measures the distribution of spatial transition in a single timestamp. Similar to density error, we use JSD to calculate the difference between  $\mathcal{T}_{syn}$ and $\mathcal{T}_{orig}$.
    \item \textbf{Pattern F1.} Despite of the single-step transition, we also consider mobility patterns with high-order dependencies. Specifically, a pattern $P$ is defined as an ordered sequence of consecutive cells. To evaluate the pattern utility in a time period of size $\varphi$, we select top-$N$ most frequent patterns in $\mathcal{T}_{syn}$ and $\mathcal{T}_{orig}$, and calculate the F1 score as the similarity measure. We report the average result of 100 random time periods and $N$ is set to 100.
\end{itemize}

\mypara{Historical Metrics}
Though \mymethod primarily focuses on real-time analysis, it is also capable of handling tasks on released historical data. Given that both global and semantic level utility are thoroughly assessed by streaming metrics, we focus on trajectory-level metrics in historical data, which is implemented on the entire trace of users rather than individual points or slices. 
Specifically, we follow~\cite{vldb23_ldptrace,ccs18_adatrace} to use \textbf{Kendall's Tau Coefficient}, \textbf{Trip Error}, and \textbf{Length Error} as our utility metrics. Kendall-tau models the discrepancies in locations' popularity ranking. Trip error and length error use JSD to measure the difference between start/end points and travel distance distribution in $\mathcal{T}_{orig}$ and $\mathcal{T}_{syn}$.

Importantly, although our global mobility model can be used to estimate certain metrics such as density error, it is unable to derive metrics involving the entire trajectory. This underscores the significance of our synthesis-based strategy, which supports arbitrary downstream tasks without consuming any additional privacy budget and avoids designing meticulous DP mechanisms tailored to each specific task. 


\begin{table*}[!t]
    \centering
    \vspace{-8mm}
    \caption{Impact of significant transition selection and entering/quitting events. The best result is shown in bold. 
    } 
    \vspace{-3mm}
    \label{tab:ablation}
    \resizebox{0.92\textwidth}{!}{
    \begin{tabular}{c|c|c|c|c|c|c|c|c|c}
    \toprule
      \textbf{Dataset} &\textbf{Model}  &\textbf{Density Error} &\textbf{Query Error} &\textbf{Hotspot NDCG}&\textbf{Transition Error}&\textbf{Pattern F1}&\textbf{Kendall Tau}&\textbf{Trip Error}&\textbf{Length Error}\\ \midrule
      \multirow{6}*{\textbf{T-Drive}}&\nosampling&0.1557&0.5171&0.4135&0.4798&0.3542&0.6952&0.3294&0.1693\\
      &\nosamplingu&0.1525&0.5130&0.4528&0.4691&0.3951&\B{0.7492}&\B{0.3208}&\B{0.1651}\\
      &\noeq&0.1413&0.8718&0.3821&0.4131&0.3818&-0.5905&0.3778&0.6931\\
      &\noequ&0.1366&0.8461&0.3935&0.4111&0.3992&-0.6286&0.3947&0.6931\\
      &\mymethodb &0.1354&0.5098&\B{0.4824}&0.4104&0.3898&0.6952&0.3362&0.2013\\ 
      &\mymethodu &\B{0.1338}&\B{0.4851}&0.4697&\B{0.3951}&\B{0.4093}&0.6984&0.3227&0.1915\\ \midrule
      \multirow{6}*{\textbf{Oldenburg}}&\nosampling
&0.1474
&0.6020&0.4595&0.5252
&0.4043&0.7381
&0.2875
&\B{0.4874}\\
      &\nosamplingu
&0.1317
&0.5688&0.5561&0.4629
&0.4436&0.7000
&0.2963
&0.4875
\\
      &\noeq
&0.1286
&0.7429&0.4428&0.4784
&0.4077&0.5381
&0.3110
&0.6931
\\
      &\noequ
&\B{0.1139}&0.7178&0.5518&0.4228&0.4492&0.6143
&0.2997
&0.6931
\\
      &\mymethodb
&0.1260&0.5958&0.4613&0.4745&0.4185&0.7540&0.2935&0.5197\\ 
      &\mymethodu &0.1171&\B{0.5629}&\B{0.5908}&\B{0.4223}&\B{0.4596}&\B{0.7635}&\B{0.2860}&0.5230\\ \midrule
      \multirow{6}*{\textbf{SanJoaquin}}&\nosampling
&0.1773
&0.5582&0.5560&0.5246
&0.4108&0.6619
&0.3630
&0.4550
\\
      &\nosamplingu
&0.1443
&0.5235&0.7738&0.4320
&0.4503&0.6587
&0.3533
&\B{0.4546}\\
      &\noeq
&0.1663
&0.8251&0.5416&0.4879
&0.4236&-0.6079
&0.3600
&0.6931
\\
      &\noequ
&0.1450
&0.8210&0.7599&0.4161
&0.4542&-0.3540
&0.3745
&0.6931
\\
      &\mymethodb 
&0.1636&0.5467&0.5641&0.4827&0.4304&0.6524&0.3602&0.5063\\ 
      &\mymethodu &\B{0.1435}&\B{0.5110}&\B{0.7913}&\B{0.4134}&\B{0.4600}&\B{0.6778}&\B{0.3499}&0.4857\\ 
      \bottomrule

    \end{tabular}}
\vspace{-15pt}
\end{table*}

\subsection{Overall Performance}\label{sec:overall_performance}
We first compare the overall utility of \mymethod and baselines with various privacy budgets. 
Due to space limits, we present the best results among our implemented allocation strategies (Adaptive, Uniform and Sample). A more detailed comparison between them is analyzed in \autoref{sec:ablation_allocation}. Based on  \autoref{tab:overall-performance}, we have the following observations:

In general, \mymethod outperforms the competitors across three datasets. Since \ldpids is originally designed for publishing static statistics,  its efficacy diminishes markedly when tasked with capturing intricate spatial-temporal patterns inherent to trajectory publishing. Though there are cases where \ldpids demonstrates commendable performance, its utility is inconsistent, underperforming in other datasets and metrics. In contrast, \mymethod delivers robust and good performance in different metrics and datasets. It also proves the versatility of \mymethod for handling various analysis tasks, spanning from streaming to historical analysis. Moreover, the significant improvement of \mymethod on trajectory-level metrics underscores the authenticity of our synthesized trajectories. This is because we effectively emulate the behavior of real-world users, thereby enhancing authenticity.

For \mymethod, population-division strategy generally outperforms budget-division strategy. This suggests that separating users rather than the budget can notably mitigate the OUE noise.  
We also find that the population-division strategy of \ldpids doesn't consistently yield superior utility. This is attributed to its assumption of a fixed active user set. When users dynamically enter and quit at each timestamp, it fails to determine the optimal allocation size, leading to suboptimal performance. In contrast, \mymethod maintains a dynamic active user set and employs a flexible portion-based allocation strategy, effectively addressing this challenge.

\begin{figure}[t]
    \centering
    \includegraphics[width=0.48\textwidth]{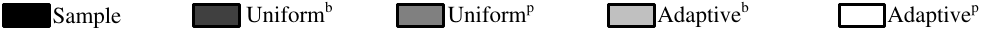}
    \includegraphics[width=0.48\textwidth]{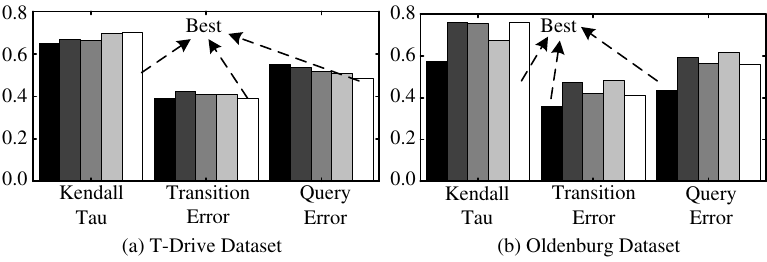}
     \vspace{-3mm}
    \caption{Impact of allocation strategy. Larger values are better for Kendall-tau and smaller values are better for Query Error and Transition Error.}
     \vspace{-6mm}
    \label{fig:allocation}
\end{figure}

We now analyze the utility \wrt different $\epsilon$s. For \mymethod, its overall performance improves with a higher $\epsilon$, due to the enhanced accuracy of OUE. 
For \ldpids, the utility fluctuates with varying $\epsilon$s. This suggests that they fail to capture the relationship between dissimilarity and perturbation noise when handling streaming dynamics of spatial-temporal patterns. 

\begin{figure*}[t]
    \centering
    \vspace{-8mm}
    \includegraphics[width=0.5\textwidth]{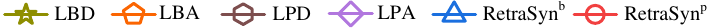}
    \includegraphics[width=0.99\textwidth]{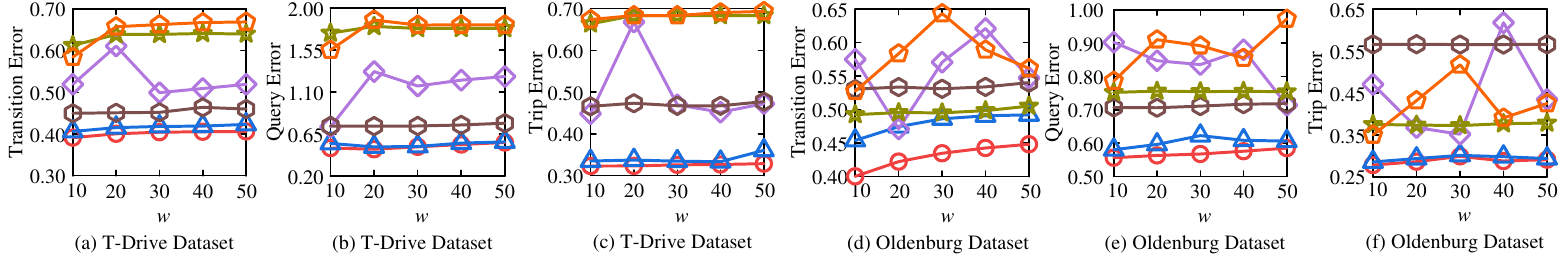} \vspace{-4mm}
    \caption{Impact of window size $w$ on T-Drive and Oldenburg dataset.} 
    \label{fig:param_w}
    \vspace{-3mm}
\end{figure*}

\begin{figure*}[t]
    \centering
    \includegraphics[width=0.5\textwidth]{figures/legends.pdf}
    \includegraphics[width=0.99\textwidth]{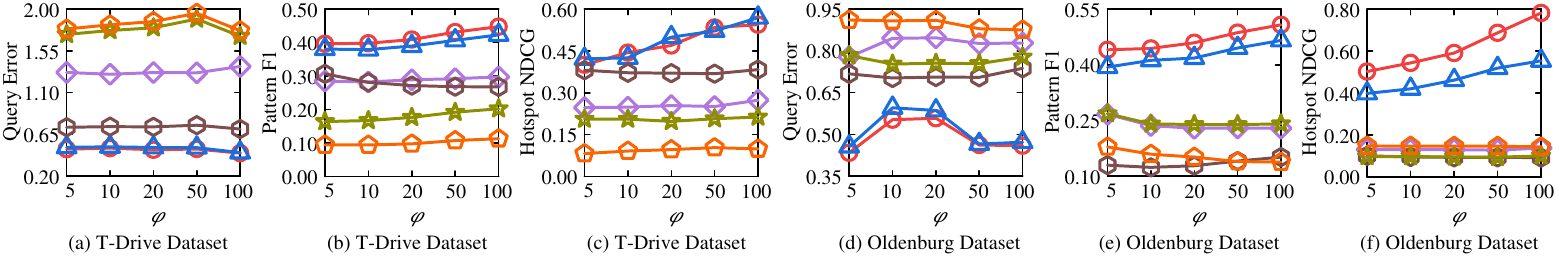} \vspace{-4mm}
    \caption{Impact of evaluation time range size $\varphi$ on T-Drive and Oldenburg dataset.}
    \label{fig:param_phi}
    \vspace{-20pt}
\end{figure*}

\subsection{Ablation Study and Parameter Variation}\label{sec:ablation}

\mypara{Impact of Significant Transition Selection}
We first verify the effectiveness of our \dmu mechanism by comparing \mymethod with variants named \nosampling and \nosamplingu. Unlike \mymethod, these variants update the entire global mobility model at each timestamp without choosing significant transitions. \autoref{tab:ablation} highlights their subpar performance in both global and semantic metrics. This indicates our \dmu can effectively prevent the model from accumulating too much perturbation noise, which reduces the overall error. 
It is also noticed that the performance of these variants sometimes matches or even surpasses \mymethod in trajectory-level metrics. One explanation is the infrequent updates of entering/quitting transitions in  \mymethod due to their typically lower frequencies, which makes them less likely to be chosen as significant transitions. This results in a delay in capturing trajectory-level information. Though our \dmu might cause a slight sacrifice in trajectory-level utility, it still ensures a more accurate global mobility model and yields better overall performance.

\mypara{Impact of Entering/quitting Events}
We also analyze the necessity of incorporating the entering and quitting events in real-world scenarios. To make a comparison, we introduce variants called \noeq and \noequ. These solutions solely contemplate normal movement between grid cells, with synthetic trajectories perpetually continuing without termination. The size adjustment process is also not used and the synthetic dataset is randomly initialized on the first timestamp. As illustrated in \autoref{tab:ablation}, the performance in some global and semantic level metrics is seemingly unaffected. This is because these metrics are largely indifferent to the nuances of entering/quitting events and predominantly rely on the overall mobility patterns that our mobility model can effectively capture. However, their performance in trajectory level metrics severely degrades compared to \mymethod. This is because without entering/quitting events,  each single synthetic trajectory fuses features of diverse real-world users. 
Moreover, the generated trajectories lack authenticity, considering their extremely long travel distance and unrealistic start/end positions.

\mypara{Impact of Allocation Strategies}  \label{sec:ablation_allocation}
We also conduct experiments on different allocation strategies, with results depicted in \autoref{fig:allocation}. 
Generally, Adaptive is more robust and achieves better overall utility with a relatively modest advantage. In T-Drive, they outperform competitors in all metrics, while in Oldenburg, Sample showcases the best transition and query error. This can be attributed to the relatively steady data changes in Oldenburg, making periodic model updates a promising strategy.  Besides, the privacy budget/users in adaptive methods are not guaranteed to be used up, leading to potentially underutilized resources. Such observations align with recent studies~\cite{vldb23_benchmark}, which find that data-independent methods suffice for certain analysis tasks. However, though the transition and query error of Sample in the Oldenburg dataset may appear favorable, its performance significantly deteriorates in terms of Kendall-tau. In comparison, the adaptive strategies consistently deliver strong performance across various tested metrics and datasets. We also evaluated the efficiency of the above strategies and the difference is negligible ($< 0.01$ seconds in T-Drive and $< 0.1$ seconds in other datasets). Therefore, adaptive methods remain a suitable choice for stable performance.

\mypara{Impact of Window Size $w$}
\autoref{fig:param_w} shows the utility under different window sizes. Due to space limits, we choose one metric from each level and demonstrate the results in T-Drive and Oldenburg datasets. In general, \mymethod consistently outperforms other strategies across all $w$s and datasets. This attests to the effectiveness of our \dmu mechanism and allocation strategy in leveraging the privacy budget and enhancing utility.  We also find that our performance displays a mild decline with a larger $w$. This is attributed to the expanded number of timestamps needing protection in a window, and the resource allocated to each timestamp decreases correspondingly. Nevertheless, \mymethod has shown its superior robustness compared to LBA and LPA, while LBD and LPD maintain a more stable performance, as their exponential decrease strategy is independent of $w$.

\mypara{Impact of Evaluation Time Range $\varphi$}
\autoref{fig:param_phi} illustrates the influence of $\varphi$. Generally, \mymethod achieves the best performance across all datasets and $\varphi$s. This indicates that \mymethod is suitable for both short-term spatial dependencies and long-term mobility patterns. For hotspot NDCG and pattern F1, the performance of \mymethod improves with a larger $\varphi$, while the competitors exhibit a relatively stable performance and even have utility degradation when $\varphi$ increases. This suggests that \mymethod achieves greater improvement when performing mid-term and long-term analysis tasks. For query error, the observed bulge in Oldenburg dataset is attributed to its skewness and the influence of the sanity bound in the evaluation function, which is designed to mitigate the influence of queries with extremely small counts~\cite{ccs18_adatrace, vldb23_ldptrace}. In a more uniform dataset such as T-Drive, it is relatively unaffected by $\varphi$.

\subsection{Efficiency and Scalability}\label{sec:scalability}
\mypara{Component Efficiency}
We first evaluate the efficiency of each sub-procedure. Specifically, we separate our framework into four processes: (i) \textit{User-side Computation}, (ii) \textit{Mobility Model Construction},  (iii) \textit{Dynamic Mobility Update} and (iv) \textit{Real-time synthesis}. 
Other operations such as privacy allocation take negligible processing time; thus, we omit them in our results. As most of the procedures in \mymethodb and \mymethodu are the same, we only report the efficiency of \mymethodu due to space limits. 
\autoref{tab:efficiency} illustrates that the synthesis takes most of the computational cost, which needs a complexity of $O(|\mathcal{T}_{syn}|)$ to generate locations and perform size adjustment. Overall, 
 our average processing time per timestamp is significantly lower than the duration of consecutive timestamps, rendering it practical for real-time analysis.

\mypara{Impact of $K$}
As analyzed in \autoref{sec:analysis}, the computational overhead is directly affected by the discretization granularity $K$. \autoref{fig:param_k} indicates that the processing time exhibits a mild increase with a larger $K$, which is attributed to the expanded grid space $\mathcal{C}$. Nevertheless, the average running time still remains reasonable and is suitable of real-time processing. The change of $K$ also affects the utility, we follow~\cite{vldb23_ldptrace} to evaluate its impact with query error. As depicted in \autoref{fig:param_k}, a larger or smaller $K$ will both lead to performance degradation. A coarser granularity will lead to uninformative mobility patterns, since many points that fall in a large region will be converted to a single cell. On the contrary, a finer granularity significantly increases the domain of transition status, perturbation on this vast domain might introduce excessive noise.

\mypara{Scalability}
We finally analyze the scalability by varying the dataset size. The average running time per timestamp is illustrated in \autoref{fig:scalability}. As observed, the processing time displays a linear growth as dataset size increases. 
Furthermore, population-division strategy exhibited slightly lower running times, primarily because only a portion of users are involved in reporting at each timestamp, thus reducing time costs.


\begin{table}[!t]
    \centering
    \caption{Component efficiency of \mymethodu. We report the average processing time per timestamp in seconds.}
    \vspace{-5pt}
    \resizebox{0.48\textwidth}{!}{
    \begin{tabular}{lccc}
    \toprule
         \textbf{Procedure} &\textbf{T-Drive} &\textbf{Oldenburg} &\textbf{SanJoaquin} \\ \midrule
         User-side Computation &0.0013 &0.0046 &0.0086 \\ 
         Mobility Model Construction &0.0002 &0.0002 &0.0002 \\
         Dynamic Mobility Update &0.0008 &0.0008 &0.0007 \\
         Real-time Synthesis &0.1828 &1.6467 &2.9463 \\ \midrule
         \textbf{Total} &\textbf{0.1851} &\textbf{1.6523} &\textbf{2.9558}\\
          \bottomrule
    \end{tabular}
    }
    \label{tab:efficiency}
    \vspace{-7mm}
\end{table}

\begin{figure}[t]
    \centering
    \includegraphics[width=0.48\textwidth]{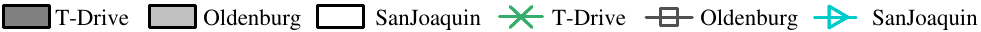}
    \includegraphics[width=0.49\textwidth]{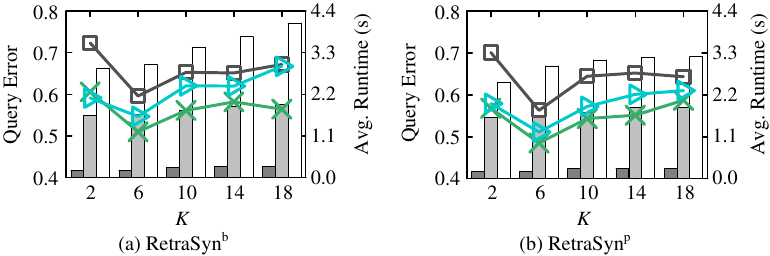}
    \vspace{-6mm}
    \caption{Impact of discretization granularity $K$. The histogram represents the running time and the lines represent utility.}
     \vspace{-4mm}
    \label{fig:param_k}
\end{figure}

\begin{figure}[t]
    \centering
    \includegraphics[width=0.22\textwidth]{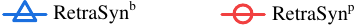}
    \includegraphics[width=0.49\textwidth]{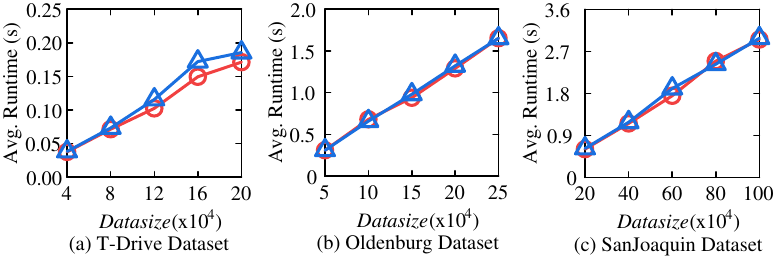}
     \vspace{-6mm}
    \caption{Scalability Evaluation.}
    \vspace{-6mm}
    \label{fig:scalability}
\end{figure}
\section{Related Work} \label{sec:related_work}
In this section, we survey the existing work for publishing general stream data and historical trajectories.  Since LDP is a variant of differential privacy (DP, also known as centralized DP), which relies on a trusted data curator to publish the sensitive data, we also introduce related work under DP.

\subsection{Stream Release with DP/LDP}
 Differential privacy has been widely applied to streaming scenarios where data is continuously released. The original privacy notion in stream publishing include \textit{event-level}~\cite{ccs17_pegasus, soda10_dwork, stoc10_dwork, icdt13_eventlevel, tissec11_eventlevel, ccs21_tops, nips18_thresh,cikm15_eventlevel,icde17_caoyang} privacy and \textit{user-level}~\cite{tkde13_fast,cikm15_dsat,cikm12_userlevel,soda19_userlevel,dbsec13_userlevel,vldb21cgm,icdew23_caoyang,sp23_userlevel} privacy. 

\mypara{Event-level Privacy}
It is designed to protect the privacy of individual timestamps. For instance, PeGaSus~\cite{ccs17_pegasus} 
utilizes a perturb-group-smooth architecture to answer multiple stream queries. 
ToPS~\cite{ccs21_tops} mitigates the problem of unbounded maximum values in real-valued data stream publication via a threshold estimation mechanism. 
The framework is applicable in both centralized and local settings of event-level privacy.
However, event-level privacy suffers from inferior privacy protection as it merely protects an individual timestamp.

\mypara{User-level Privacy}
In contrast to event-level privacy, it aims to hide all events of any user. A notable work is FAST~\cite{tkde13_fast}, which tackles user-level privacy by employing a sampling strategy to select the appropriate publishing timestamps. The sampling interval is adaptively adjusted according to the data dynamics. While user-level privacy protects the entire stream of each user, it is not feasible for infinite data streams.

\mypara{$w$-event Privacy}
To balance the privacy and utility of event-level and user-level privacy, $w$-event DP~\cite{vldb14_bdba} is proposed for infinite streams, which protects each running window of at most $w$ timestamps. The authors further present two algorithms called Budget Distribution (BD) and Budget Absorption (BA) to adaptively allocate the privacy budget to the selected timestamps for perturbation. 
Researchers have further explored the application of $w$-event DP in the centralized setting~\cite{tdsc18_rescuedp,ipccc19_adapub,cikm15_dsat,apscc16_gevent,securecomm16_secweb,tvt19-rptr,tmc18_wevent,rgp_wevent}. One notable example is RescueDP~\cite{tdsc18_rescuedp}, which builds upon the FAST framework. Moreover, RescueDP incorporates a grouping strategy to group the dimensions with similar statistics and changing trends, in order to reduce the noise added to statistics with small values.

Recently, the $w$-event privacy has also been applied to LDP~\cite{infocom20_patternldp,sigmod22_ldpids,bigdata18_wevent,trustcom19_wevent}. Wang \etal~\cite{infocom20_patternldp} propose a pattern-aware stream publishing model, which incorporates the correlations between consecutive data points in time series. However, it achieves a metric-based LDP rather than pure $\epsilon$-LDP. 
The state-of-the-art method based on pure $\epsilon$-LDP is LDP-IDS~\cite{sigmod22_ldpids}, which adopts the idea of BD and BA to the local setting. The authors also present two population division-based mechanisms that focus on separating participant users rather than splitting the privacy budget for each subroutine. 

Compared to existing stream release solutions, our method fully harnesses spatial-temporal context information and has superior versatility for various downstream tasks. 

\subsection{Historical Trajectory Release with DP/LDP}
Private publication of trajectory and other location data~\cite{vldb22_snh,ccs21_directional,sigmod23_neural} with DP/LDP has been a prominent research area for over a decade.
Existing DP/LDP solutions that are applicable in trajectory publication can be broadly classified into point perturbation methods~\cite{ccs13_geoindis,cikm13_pointperturb,icde22_gl, vldb21_ngram, vldb23_direction} and synthesis-based methods~\cite{sp16_sglt,kdd12_ngram,sstd21_cormode,ccs18_adatrace,gursoy2020utility,vldb15_dpt,mdm22,vldb23_ldptrace,vldb22_mtnet}.

\mypara{Point Perturbation Methods}
These methods add noise to individual spatial points in trajectory data before they are published or used for analysis tasks. For example, NGRAM~\cite{vldb21_ngram} is a framework proposed for privacy trajectory sharing under LDP. 
NGRAM leverages auxiliary external knowledge and overlapped n-grams to preserve the spatial-temporal-category information of the original data. 
ATP~\cite{vldb23_direction} combines the direction information with trajectory perturbation in LDP. 

\mypara{Synthesis Methods}
These methods focus on generating synthetic trajectories that collectively exhibit a high resemblance to the real dataset while ensuring privacy.
AdaTrace~\cite{ccs18_adatrace} extracts four key spatial features from original trajectories and incorporates them into a generative synthesizer. 
LDPTrace~\cite{vldb23_ldptrace},  the state-of-the-art trajectory synthesis framework under LDP, takes into account three crucial spatial patterns to form a probabilistic distribution, including the beginning/terminated points, intra-trajectory transitions and trajectory length. Afterwards, a Markov-based generative model performs synthesis utilizing the extracted mobility patterns.

By comparison,
\mymethod constructs and maintains a dynamic global mobility model. The \dmu mechanism and adaptive allocation strategy enable real-time processing.

\section{Conclusions} \label{sec:conclusion}
We propose \mymethod, a novel real-time trajectory synthesis framework under LDP. To ensure utility, we harness the spatial-temporal context to build a global mobility model, extracting the mobility patterns within the original streams. An optimization-based mechanism is employed to dynamically update the mobility model while reducing the introduced error.  Furthermore, we emulate the behavior of real-world travelers and explore the allocation strategy in realistic scenarios to improve authenticity. Extensive experiments are conducted and demonstrate the superiority and versatility of \mymethod. In the future, we aim to study acceleration techniques (\eg parallel computing) to further enhance the efficiency, and to integrate \mymethod into distributed trajectory management systems.

\balance
\bibliographystyle{abbrv}
\bibliography{ref}
\end{document}